\documentclass[aps,pra,reprint,superscriptaddress,amsmath,amssymb,showpacs]{revtex4-1}
\usepackage{bm}
\newtheorem{thm}{Theorem}
\newtheorem{prop}[thm]{Proposition}
\newtheorem{coro}[thm]{Corollary}
\newenvironment{proof}{\noindent{\it Proof. \hskip0pt}}
                      {$\square$\par\medskip}
\begin{document}


\title{Separability of qubit-qudit quantum states with strong positive partial transposes}

\author{Kil-Chan Ha}
\affiliation{Faculty of Mathematics and Statistics, Sejong University, Seoul 143-747, Korea}
\date{\today}

\begin{abstract}
We show that all $2\otimes 4$ states with strong positive partial transposes (SPPT) are separable. We also construct a family of $2\otimes 5$ entangled SPPT states, so the conjecture on the separability of SPPT states are completely settled. In addition, we clarify the relation between the set of all $2\otimes d$ separable states and the set of all $2\otimes d$ SPPT states for the case of $d=3,4$.
\end{abstract}


\maketitle


The notion of quantum entanglement plays a key role in the current study of quantum information and quantum computation theory. One of the central problems in the theory of quantum entanglement is to check whether a given density matrix representing a quantum state of composite system is separable or entangled. Let us recall that a state $\rho$ acting on the Hilbert space $\mathcal H_A\otimes \mathcal H_B$ is called separable if it is a convex combination of product states, that is, $\rho=\sum_k p_k\rho_k\otimes \tilde \rho_k$, where $\rho_k$ and $\tilde \rho_k$ are states acting on $\mathcal H_A$ and $\mathcal H_B$, respectively \cite{werner}.

There are two main criteria for separability: One criterion was given by Horodecki {\it et al.} \cite{horo-1} using positive linear maps between matrix algebras, and this was formulated as the notion of entanglement witnesses \cite{terhal}. Another criterion, the so called partial positive transpose (PPT) criterion \cite{peres} tells us that if a state $\rho$ is separable then its partial transposition $\rho^{T_A}=(T\otimes \openone)\rho$ is positive. It was shown by Horodecki {\it et al.} \cite{horo-2} that the PPT criterion is also a sufficient condition for separability in the system $2\otimes 2$ and $2\otimes 3$. For higher dimensions, this is not the case by a work of Woronowicz \cite{woro} who gave an example of a $2\otimes 4$ PPT entangled state. Such examples were also given in Refs. \cite{choi, stor} for the $3\otimes 3$ cases, in the early eighties. See also Ref. \cite{phoro} for the $2\otimes 4$ case. Therefore, it is important to understand which PPT states are separable and which are entangled.

In this context, a subclass of PPT states, whose PPT properties are ensured by the canonical construction using Cholesky decomposition, were considered in Ref.~\cite{chrus}. Such states are called strong positive partial transpose (SPPT) states. Based on several examples of SPPT states, it was conjectured in \cite{chrus} that all SPPT states are separable. Unfortunately this is not true in general $m\otimes n$ SPPT states with $m,\,n >2$ since there exist entangled $3\otimes 3$ states which are SPPT \cite{ha10}. However, this conjecture is still open for $2\times d$ SPPT states. We note that $2\otimes d$ systems are particularly useful since it allows us to determine a number of separability properties in the multiqubit case \cite{dct}, and were intensively analyzed in \cite{kckl}. See also Ref. \cite{chrus2} in which it was proved that SPPT notion can be used for witnessing quantum discord in $2\otimes d$ systems.

In this Brief Report, we show that all $2\otimes 4$ SPPT states are separable, but  the conjecture proposed in \cite{chrus} does not hold true for general $2\otimes d$ systems with $d\ge 5$. This result displays the difference between  the $2\otimes d$ case and the general $m\otimes n$ one ($m,n>2$), even difference between $2\otimes 4$ and $2\otimes 5$ cases. We also clarify the relation between the notion of separability and the SPPT property.

We begin with the definition of a $2\otimes d$ SPPT state. Consider the following class of upper triangular block matrices $\mathbf X$ and $\mathbf Y$:
\begin{equation}\label{eq:tri_block}
\mathbf X =\begin{pmatrix} X_1 & S X_1\\ 0 & X_2\end{pmatrix},\quad
\mathbf Y =\begin{pmatrix} X_1 & S^{\dagger} X\\ 0 & X_2\end{pmatrix},
\end{equation}
where $X_k$ and $S$ are arbitrary $d\times d$ matrices. Then we say that a state
\begin{equation}\label{sppt_state}
 \rho=\mathbf{X}^{\dagger}\mathbf{X}
=\begin{pmatrix} 
X_1^{\dagger}X_1 & X_1^{\dagger}S X_1\\ X_1^{\dagger}S^{\dagger}X_1 & X_1^{\dagger}S^{\dagger}S X_1+X_2^{\dagger}X_2
\end{pmatrix}
\end{equation}
with $\mathbf{X}$ defined in \eqref{eq:tri_block} has SPPT if 
$\rho^{T_A}=\mathbf Y^{\dagger} \mathbf Y$ with $\mathbf Y$ defined in \eqref{eq:tri_block}. It is clear that $\rho=\mathbf X^{\dagger}\mathbf X$ has SPPT if and only if
\begin{equation}\label{cond:sppt}
X_1^{\dagger}S^{\dagger}SX_1=X_1^{\dagger}SS^{\dagger}X_1.
\end{equation}
We also note that $\rho$ has SPPT if and only if $(\openone \otimes V)^{\dagger}\rho (\openone\otimes V)$ has SPPT with nonsingular $d\times d$ matrix.

First, we consider $2\otimes d$ SPPT state $\rho$ with $r(X_1^{\dagger}X_1)=d$, where $r(X)$ denotes the rank of $X$. In this case, $X_1$ is nonsingular, and so $S$ is normal matrix by the condition \eqref{cond:sppt}. Thus we have the spectral decomposition for $S$
\[
S=\sum_{i=1}^d \lambda_i P_i,
\]
where $P_i$'s are rank one projections with $\sum_{i=1}^d P_i=\openone$. 
Then we can write 
\[
\begin{pmatrix} \openone & S\\S^{\dagger} & S^{\dagger}S\end{pmatrix}
=\sum_{i=1}^d \sigma_i \otimes P_i\quad {\text with }\ 
\sigma_i=\begin{pmatrix} 1 & \lambda_i\\ \lambda_i^* & |\lambda_i|^2\end{pmatrix}
\]
and we see that
\[
\begin{pmatrix} X_1^{\dagger} & 0 \\ 0 & X_1^{\dagger}\end{pmatrix}
\begin{pmatrix} \openone & S\\S^{\dagger} & S^{\dagger}S\end{pmatrix}
\begin{pmatrix} X_1 & 0 \\ 0 & X_1 \end{pmatrix}
=
\sum_{i=1}^d (\sigma_i \otimes X_1^{\dagger}P_i X_1)
\]
is an unnormalized separable state.
Therefore, we can conclude that
\[
\rho=\begin{pmatrix} X_1^{\dagger} & 0 \\ 0 & X_1^{\dagger}\end{pmatrix}
\begin{pmatrix} \openone & S\\S^{\dagger} & S^{\dagger}S\end{pmatrix}
\begin{pmatrix} X_1 & 0 \\ 0 & X_1 \end{pmatrix}+
\begin{pmatrix} 0 & 0 \\ 0 & X_2^{\dagger}X_2\end{pmatrix}
\]
is separable since the second matrix in the righthand side is $|1\rangle \langle 1|\otimes X_2^{\dagger}X_2$, where $|1\rangle =(0,1)^{\rm t}$. 
Consequently we have the following.
\begin{prop} Let $\rho$ be a $2\otimes d$ SPPT state of the form \eqref{sppt_state} with $r(X_1^{\dagger}X_1)=d$. Then $\rho$ is separable.
\end{prop}

From the above Proposition, we obtain a sufficient condition for separability of $2\otimes d$ PPT states $\rho$ with $r(\langle 0|\rho|0\rangle)=d$. To see this, we observe the condition when such a PPT state is SPPT. Let $\rho$ be a $2\otimes d$ PPT state of the form 
\begin{equation}\label{form:ppt}
\rho=\begin{pmatrix} A & B\\B^{\dagger} & C\end{pmatrix}\ \text{ with } r(A)=d.
\end{equation}
From the PPT property of $\rho$, we see that  both $C-B^{\dagger}A^{-1} B$ and $C-B A^{-1}B^{\dagger}$ are positive semi-definite matrices (see the Theorem 1.3.3 in Ref.~\cite{bha}). So we can find $X_2$ satisfying the following condition
\[
C-B^{\dagger}A^{-1}B=X_2^{\dagger} X_2
\]
since $A$ is a invertible positive definite matrix.
Thus we have 
\[
\rho=\begin{pmatrix} A & B\\B^{\dagger} & C\end{pmatrix}
=\begin{pmatrix} X_1^{\dagger}X_1 & X_1^{\dagger} \widetilde B X_1\\
X_1^{\dagger}\widetilde B^{\dagger} X_1 & X_1^{\dagger}\widetilde B^{\dagger}\widetilde B X_1 +X_2^{\dagger}X_2
\end{pmatrix},
\]
where $X_1=A^{1/2}=(A^{\dagger})^{1/2}$, $\widetilde B=(A^{\dagger})^{-1/2}B A^{-1/2}$.
Consequently, we have the following result from the condition \eqref{cond:sppt}. 
\begin{coro}\label{coro1}Let $\rho$ be a $2\otimes d$ PPT state of the form~\eqref{form:ppt}. Then $\rho$ is SPPT if and only if the condition $B^{\dagger}A^{-1}B=B A^{-1}B^{\dagger}$ is satisfied. In this case, the PPT state $\rho$ is separable.
\end{coro}

Now, we consider a $2\otimes d$ SPPT state $\rho$ with $r(X_1^{\dagger}X_1)=k<d$. Then we may write $X_1=U\Sigma V^*$ for some $d\times d$ unitary matrices $U,\,V$ and a diagonal matrix $\Sigma$ of rank $k$ with diagonal entries $\sigma_{11}\ge \sigma_{22}\ge \cdots \ge \sigma_{kk}\ge 0$ by the singular value decomposition. Thus $\rho$ in \eqref{sppt_state} can be written by 
\begin{equation}\label{sppt_dec}
\begin{aligned}
\rho=&\begin{pmatrix} 
X_1^{\dagger}X_1 & X_1^{\dagger}S X_1\\ X_1^{\dagger}S^{\dagger}X_1 & X_1^{\dagger}S^{\dagger}S X_1 +X_2^{\dagger}X_2
\end{pmatrix}\\
=&\begin{pmatrix}V & 0\\0 & V\end{pmatrix} 
\begin{pmatrix} \Sigma^2 & \Sigma \widetilde{S} \Sigma\\ \Sigma \widetilde{S}^{\dagger}\Sigma & \Sigma \widetilde{S}^{\dagger}\widetilde{S}\Sigma\end{pmatrix}
\begin{pmatrix}V^{\dagger} & 0\\ 0 & V^{\dagger}\end{pmatrix} +
\begin{pmatrix}0 & 0 \\0 & X_2^{\dagger}X_2\end{pmatrix},
\end{aligned}
\end{equation}
where $\widetilde S=U^{\dagger} S U$.
Now, we write $\Sigma$ and $\widetilde S$ as block matrices
\[
\Sigma=\begin{pmatrix} D_k & 0 \\ 0 & 0\end{pmatrix},\quad \widetilde S=\begin{pmatrix} \widetilde{S}_{11} & \widetilde{S}_{12}\\\widetilde{S}_{21} &\widetilde{S}_{22}\end{pmatrix},
\] 
where $D_k$ and $\widetilde S_{11}$ are $k\times k$ matices. Then we have 
\[
\begin{aligned}
\widetilde \rho:=&\begin{pmatrix} \Sigma^2 & \Sigma \widetilde{S} \Sigma\\ \Sigma \widetilde{S}^{\dagger}\Sigma & \Sigma \widetilde{S}^{\dagger}\widetilde{S}\Sigma\end{pmatrix}\\
=&\begin{pmatrix} D_k^2 & 0 & D_k \widetilde S_{11} D_k & 0\\ 0 & 0 & 0 & 0\\D_k \widetilde S^{\dagger}_{11}D_k & 0 & D_k (\widetilde S^{\dagger}_{11}\widetilde S_{11}+\widetilde S^{\dagger}_{21}\widetilde S_{21})D_k & 0 \\ 0 & 0 & 0 & 0\end{pmatrix},
\end{aligned}
\]
with $D_k(\widetilde S_{11}^{\dagger}\widetilde S_{11}+\widetilde S_{21}^{\dagger}\widetilde S_{21})D_k=D_k(\widetilde S_{11}\widetilde S_{11}^{\dagger}+\widetilde S_{12}\widetilde S_{12}^{\dagger})D_k$ by the condition~\eqref{cond:sppt}.
It is easy to see that $\widetilde \rho$ is unnormalized separable state if and only if the following reduced unnormalized $2\otimes k$ state
\begin{equation}\label{eq:ppt}
\begin{pmatrix}
D_k^2 & D_k\widetilde S_{11} D_k\\ D_k\widetilde S^{\dagger}_{11} D_k & D_k(\widetilde S^{\dagger}_{11}\widetilde S_{11}+\widetilde S_{21}^{\dagger}\widetilde S_{21})D_k\end{pmatrix}
\end{equation}
is separable. We note that the above $2\otimes k$ state is a PPT state, although it may not be SPPT. Therefore, if $k\le 3$ then we see that $\widetilde \rho$ is separable. In this case, we can conclude that $\rho$ is separable in \eqref{sppt_dec}. Consequently, we have the following.
\begin{prop} 
Let $\rho$ be a $2\otimes d$ SPPT state of the form \eqref{sppt_state} with $r(X_1^{\dagger}X_1)\le 3$. Then $\rho$ is separable.
\end{prop}
To answer the conjecture asked in \cite{chrus}, we will show the following.
\begin{thm}
All $2\otimes d$ states with strong positive partial transpose are separable if and only if $d\le 4$.
\end{thm}
\begin{proof}
By combining the Proposition 1 and 2, we see that all $2\otimes 4$ SPPT states are separable. To complete the proof, we construct a $2\otimes 5$ SPPT state which is not separable.
Define $5\times 5$ matrices $X_1$ and $S$ by 
\[
X_1=\begin{pmatrix} 1 & 0 & 0 & 0 & 0 \\0 & 1 & 0 & 0 & 0 \\ 0 & 0 & 1 & 0 & 0\\ 0 & 0 & 0 & 1 & 0\\0 & 0 & 0 & 0 & 0\end{pmatrix},\
S=\begin{pmatrix} 0 & 1 & 0 & 0 & \beta_1 \\ 0 & 0 & 1 & 0 & 0 \\ 0 & 0 & 0 & 1 & 0\\0 & 0 & 0 & 0 & \beta_2\\ \beta_2 & 0 & 0 & \beta_1 & 0\end{pmatrix}
\]
where $\beta_1=[(1-b)/2b]^{\frac 12}$ and $\beta_2=[(1+b)/2b]^{\frac 12}$ with $0<b<1$. We also put $X_2$ by $5\times 5$ zero matrix. Then we define $2\otimes 5$ SPPT state $\varrho_0$ by Eqs.~\eqref{eq:tri_block} and \eqref{sppt_state}:
\[
\varrho_0=\begin{pmatrix}
1 & 0 & 0 & 0 & \cdot & 0 & 1 & 0 & 0 & \cdot\\
0 & 1 & 0 & 0 & \cdot & 0 & 0 & 1 & 0 & \cdot\\
0 & 0 & 1 & 0 & \cdot & 0 & 0 & 0 & 1 & \cdot\\
0 & 0 & 0 & 1 & \cdot & 0 & 0 & 0 & 0 & \cdot\\
\cdot & \cdot & \cdot & \cdot & \cdot & \cdot & \cdot & \cdot & \cdot & \cdot  \\
0 & 0 & 0 & 0 & \cdot & \gamma_1 & 0 & 0 & \gamma_2 & \cdot\\
1 & 0 & 0 & 0 & \cdot & 0 & 1 & 0 & 0 & \cdot\\
0 & 1 & 0 & 0 & \cdot & 0 & 0 & 1 & 0 & \cdot\\
0 & 0 & 1 & 0 & \cdot & \gamma_2 & 0 & 0 & \gamma_1 & \cdot\\
\cdot & \cdot & \cdot & \cdot & \cdot & \cdot & \cdot & \cdot & \cdot & \cdot  
\end{pmatrix},
\]
where $\gamma_1=(b+1)/(2b),\ \gamma_2=\sqrt{b^2-1}/(2b)$, and dot($\cdot$) denotes zero. We note that the corresponding reduced $2\otimes 4$ state as in \eqref{eq:ppt} is the PPT entangled state given by Horodecki \cite{phoro}. Therefore, we conclude that $\varrho_0$ is an entangled state with strong positive partial tranpose. This completes the proof.
\end{proof}

Lastly, for $d=3,4$, we show that the set of all $2\otimes d$ SPPT states is proper subset of the set of all separable $2\otimes d$ states. To see this, we consider a $2\otimes 3$ state $\varrho_1=\begin{pmatrix}
\rho_{11} & \rho_{12}\\\rho_{12}^{\dagger} & \rho_{22}\end{pmatrix}$ defined by 
\[
\rho_{11}=\begin{pmatrix}
3 & \cdot & \cdot \\
\cdot & 4 & 2 \\
\cdot  & 2 & 3\end{pmatrix},\,
\rho_{12}=\begin{pmatrix}
\cdot & \cdot & \cdot\\
\cdot & \cdot & 1\\
1 & -1 & \cdot\end{pmatrix},\,
\rho_{22}=\begin{pmatrix}
2 & 1 & -1\\
1 & 6 & 1\\
-1 & 1 & 3\end{pmatrix}.
\]
Then, we can easily check that  four $3\times 3$ matrices $\rho_{11},\, \rho_{22},\, \rho_{22}-\rho_{12}^{\dagger}\rho_{11}^{-1}\rho_{12}$ and $ \rho_{22}-\rho_{12} \rho_{11}^{-1} \rho_{12}^{\dagger}$ are all positive definite matrices. Therefore, we see that $\varrho_1$ is $2\otimes 3$ PPT state, and so $\varrho_1$ is separable. On the other hand, we see that
\[
\rho_{12}^{\dagger}\rho_{11}^{-1}\rho_{12}-\rho_{12} \rho_{11}^{-1}\rho_{12}^{\dagger}=\dfrac 1{12}\begin{pmatrix} 6 & -6 & -3\\-6 & 0 & 0\\-3 & 0 & -4\end{pmatrix}.
\]
Therefore, $\varrho_1$ is not SPPT by the Corollary~\ref{coro1}.
Now, we define $2\otimes 4$ state $\rho_2$ using the above $2\otimes 3$ state $\rho_1$ as follows:
\[
\rho_2=\begin{pmatrix}
\rho_{11} & 0 & \rho_{12} & 0\\
0 & 1 & 0 & 0\\
\rho_{12}^{\dagger} & 0 & \rho_{22} & 0\\
0 & 0 & 0 & 0\end{pmatrix}.
\]
Then it is obvious that $\varrho_2$ is separable. We can also show that $\varrho_2$ is not SPPT by the Corollary~\ref{coro1}. This completes the proof of claim.

In conclusion, we showed that all $2\otimes 4$ SPPT states are separable. We also construced a family of $2\otimes 5$ SPPT entangled states using Horodecki's $2\otimes 4$ PPT entangled states. So the conjecture \cite{chrus} on the separability of SPPT states is completely settled. We also clarify the relation between the set of all $2\otimes d$ separable states and the set of all $2\otimes d$ SPPT states for the case of $d=3,4$.

\begin{acknowledgments}
This work was partially supported by the Basic Science Research Program through the
National Research Foundation of Korea(NRF) funded by the Ministry of Education, Science
and Technology (Grant No. NRFK 2012-0002600)
\end{acknowledgments}


\begin{thebibliography}{99}
\bibitem{werner} R. F. Werner, 
Phys. Rev. A {\bf 40}, 4277 (1989).

\bibitem{horo-1} M. Horodecki, P. Horodecki, and R. Horodecki,
Phys. Lett. A {\bf 223}, 1 (1996).

\bibitem{terhal} B. M. Terhal,
Phys. Lett. A {\bf 271}, 319 (2000).

\bibitem{peres} A. Peres,
Phys. Rev. Lett. {\bf 77}, 1413 (1996).

\bibitem{horo-2} P. Horodecki, 
Phys. Lett. A {\bf 232}, 333 (1997).

\bibitem{woro} S. L. Woronowicz,
Rep. Math. Phys. {\bf 10}, 165 (1976).

\bibitem{choi} M.-D. Choi,
{\it Proceedings of Symposia in Pure Mathematics} (Providence, RI: American Mathematical Society, 1982), Vol. 38, Part 2, pp. 583-90

\bibitem{stor} E. St\o rmer,
Proc. Am. Math. Soc. {\bf 86}, 402 (1982).

\bibitem{phoro} P. Horodecki, 
Phys. Lett. A {\bf 232}, 333 (1997).

\bibitem{chrus} D. Chru\'sci\'nski, J. Jurkowski, and A. Kossakowski, 
Phys. Rev. A {\bf 77}, 022113 (2008).

\bibitem{ha10} K.-C. Ha,
Phys. Rev. A {\bf 81}, 064101 (2010).

\bibitem{dct} W. D\"{u}r, J. I. Cirac, and R. Tarrach,
Phys. Rev. Lett. {\bf 83}, 3562 (1999).

\bibitem{kckl} B. Kraus, J. I. Cirac, S. Karnas, and M. Lewenstein, 
Phys. Rev. A {\bf 61}, 062302 (2000).

\bibitem{chrus2} B. Bylicka and D. Chru\'sci\'nski,
Phys. Rev. A {\bf 81}, 062102 (2010).

\bibitem{bha} R. Bhatia,
{\it Positive Definite Matrices} (Princeton University Press, New Jersey, 2007), p. 14


\end{thebibliography}

\end{document}